\documentclass[a4paper,UKenglish]{lipics-v2019}
\usepackage{amsmath,amsfonts,latexsym}
\usepackage{graphicx}
\usepackage{float}
\usepackage{enumerate}
\usepackage{multirow}
\usepackage{mathtools}
\usepackage{amssymb}
\usepackage{braket}
\usepackage{psfrag, epstopdf}
\usepackage{cite}
\usepackage{algorithm2e}
\usepackage[noend]{algorithmic}
\usepackage{tabu}

\newcommand{\N}{{\mathbb{N}}}  
\newcommand{\Z}{{\mathbb{Z}}}  
\newcommand{\Q}{{\mathbb{Q}}} 
\newcommand{\R}{{\mathbb{R}}}  
\newcommand{\C}{{\mathbb{C}}}  
\newcommand{\pfa}{{\mathcal{P}}}  

\newcommand{\0}{\mathbf{0}} 
\theoremstyle{plain}
\newtheorem{thm}{Theorem}
\newtheorem{lem}[thm]{Lemma}
\newtheorem{prob}[thm]{Problem}
\newtheorem{cor}[thm]{Corollary}
\newtheorem{prop}[thm]{Proposition}

\newtheorem{innercustomcor}{Corollary}
\newenvironment{customcor}[1]
  {\renewcommand\theinnercustomcor{#1}\innercustomcor}
  {\endinnercustomcor}

\newlength\myindent
\newcommand\bindent{%
  \begingroup
  \setlength{\itemindent}{\myindent}
  \addtolength{\algorithmicindent}{\myindent}
}
\newcommand\eindent{\endgroup}

\title{Decidability of cutpoint isolation for probabilistic finite automata on letter-bounded inputs}

\titlerunning{Decidability of cutpoint isolation for PFA on letter-bounded inputs}

\author{Paul C. Bell}{Department of Computer Science, James Parsons Building, Byrom Street, Liverpool John Moores University, Liverpool, L3 3AF, UK}{p.c.bell@ljmu.ac.uk}{https://orcid.org/0000-0003-2620-635X}{}

\author{Pavel Semukhin}{Department of Computer Science, University of Oxford, Wolfson Building, Parks Road, Oxford, OX1 3QD, UK}{pavel.semukhin@cs.ox.ac.uk}{https://orcid.org/0000-0002-7547-6391}{}

\authorrunning{P.\,C. Bell and P. Semukhin}

\Copyright{Paul C. Bell and Pavel Semukhin}



\nolinenumbers 
\hideLIPIcs  

\begin{document}

\maketitle

\begin{abstract}
We show the surprising result that the cutpoint isolation problem is decidable for probabilistic finite automata where input words are taken from a letter-bounded context-free language. A context-free language $\mathcal{L}$ is letter-bounded when $\mathcal{L} \subseteq a_1^*a_2^* \cdots a_\ell^*$ for some finite $\ell > 0$ where each letter is distinct. A cutpoint is isolated when it cannot be approached arbitrarily closely. The decidability of this problem is in marked contrast to the situation for the (strict) emptiness problem for PFA which is undecidable under the even more severe restrictions of PFA with polynomial ambiguity, commutative matrices and input over a letter-bounded language as well as to the injectivity problem which is undecidable for PFA over letter-bounded languages. We provide a constructive nondeterministic algorithm to solve the cutpoint isolation problem, which holds even when the PFA is exponentially ambiguous. We also show that the problem is at least NP-hard and use our decision procedure to solve several related problems.
 \end{abstract}
 
 \section{Introduction}

Probabilistic finite automata (PFA) are an extension of classical nondeterministic finite automata (NFA) where transitions, for each state and letter, are represented as probability distributions. The PFA model was first introduced by Rabin \cite{Ra63}. 

There are a variety of classical problems for PFA. Let $\pfa$ denote a PFA, $\Sigma$ an alphabet and $\lambda \in [0, 1]$ a probability. The acceptance probability of $\pfa$ on a word $w \in \Sigma^*$ is denoted $f_{\mathcal{P}}(w)$. A central question is \emph{(strict) emptiness} of cutpoint languages: does there exist a finite input word $w$ for which $f_{\mathcal{P}}(w) \geq \lambda$ (or $f_{\mathcal{P}}(w) > \lambda$ for strict emptiness).  Another important problem is that of \emph{cutpoint isolation} --- to determine if $\lambda$ can be approached arbitrarily closely, i.e., for each $\epsilon > 0$, does there exist a word $w \in \Sigma$ such that $|f_{\mathcal{P}}(w) - \lambda| < \epsilon$ (or the converse, does there exist $\delta>0$ such that $|f_{\pfa}(w) - \lambda| \geq \delta$ for all $w \in \Sigma^*$)? The \emph{value}-$1$ problem is a special case of the cutpoint isolation when $\lambda = 1$ \cite{GO10}. In the \emph{injectivity problem} we must determine if $f_{\mathcal{P}}(w)$ is injective (i.e. do there exist two distinct words with the same acceptance probability?) In the \emph{$\lambda$-probability problem} we must determine if  there exist $w \in \Sigma^*$ such that $f_{\mathcal{P}}(w) = \lambda$. 

The emptiness problem is undecidable for rational matrices \cite{Paz}, even over a binary alphabet when the PFA has dimension $46$ \cite{BC03}, later improved to dimension $25$ \cite{Hir}. The injectivity problem for PFA is undecidable \cite{BCJ19}, even for polynomially ambiguous PFA \cite{Bell19}. 

The main focus of this paper is the cutpoint isolation problem. The authors of \cite{BMT77} show that the problem of determining if a given cutpoint is isolated (resp. if a PFA has any isolated cutpoint) is undecidable and this was shown to hold even for PFA with $420$ (resp. $2354$) states over a binary alphabet \cite{BC03}. The cutpoint isolation problem, in the special case where $\lambda = 1$ (the value-$1$ problem), is also known to be undecidable \cite{GO10}. The problem is especially interesting given the seminal result of Rabin that if a cutpoint $\lambda$ is isolated, then the cutpoint language associated with $\lambda$ is necessarily regular \cite{Ra63}. 


Most problems are undecidable for PFA and there exist very few algorithmic solutions \cite{GO10}. Various classes of restrictions on PFA are possible, related to the number of states, the alphabet size and whether one defines the PFA over the algebraic reals or the rationals. Recent work has studied PFA with finite, polynomial or exponential ambiguity (in terms of the underlying NFA) \cite{FR17}, PFA defined for restricted input words (e.g. those coming from bounded or letter-bounded languages) \cite{BHH13, BCJ19}, commutative PFA, where all transition matrices commute, for which cutpoint languages and non-free languages generated by such automata become commutative \cite{Bell19} or other structural restrictions on the PFA such as \#-acyclic automata, for which some problems become decidable \cite{GO10}, including the value-$1$ problem. Such \#-acyclic automata impose a restriction on the structure of the PFA (as we shall see, we only restrict the input words).

A natural restriction on PFA was studied in \cite{BHH13}, where input words of the PFA are restricted to be from a letter-bounded language (also known as a \emph{letter-monotonic language}) of the form $\mathcal{L} = a_1^*a_2^* \cdots a_\ell^*$ with distinct letters $a_i \in \Sigma$. This is analogous to a $1.5$-way PFA, whose read head may ``stay put'' on an input letter but never moves left. This may model a situation where we have some finite number of probabilistic events and we know that there is a fixed order and number of transitions between them, but with each event being applied an arbitrary number of times. The model is also related to ``promise problems'' whereby we restrict the decision question to a subset of possible inputs \cite{Go06}. Letter-bounded languages allow a natural and substantial extension to decision questions on a unary alphabet.

The emptiness and $\lambda$-probability problems for PFA on letter-bounded languages were shown to be undecidable for high (finite) dimensional matrices via an encoding of Hilbert's tenth problem on the solvability of Diophantine equations and Turakainen's method to transform weighted integer automata to probabilistic automata \cite{Tu69}. These undecidability results also hold for polynomially ambiguous PFA with commutative matrices \cite{Bell19}.

The authors of \cite{FR17} studied decision problems for PFA of various degrees of ambiguity. The degree of ambiguity (finite, polynomial or exponential) of a PFA is a structural property, giving an indication of the number of accepting runs for a given input word. The degree of ambiguity of automata is a well-known and well-studied property in automata theory \cite{WS91}. The authors of \cite{FR17} show that the emptiness problem for PFA remains undecidable even for polynomially ambiguous automata (quadratic ambiguity), show {\bf PSPACE}-hardness results for finitely ambiguous PFA and that emptiness is in {\bf NP} for the class of $k$-ambiguous PFA for every $k > 0$. The emptiness problem for PFA was later shown to be undecidable for linearly ambiguous automata \cite{DJ18}. 

\subsection{Our Contributions}

It is natural to consider the decidability of the cutpoint isolation problem for polynomially ambiguous PFA on letter-bounded or commutative languages, given that the (strict) emptiness problems for such automata are undecidable \cite{Bell19}. In the present paper we prove the surprising result that the cutpoint isolation problem is in fact \emph{decidable}, even if the PFA is exponentially ambiguous, matrices are non-commutative, and the input language is not just the letter-bounded language $a_1^* \cdots a_\ell^*$ but instead a more general letter-bounded  context-free language. The results are shown in Table~\ref{results1Tab}.

\begin{table}[h]
\footnotesize
\begin{center}
{\tabulinesep=1.2mm
\begin{tabu}{|c|c|c|c|}
\hline
\multirow{3}{*}{Problem} & \multirow{3}{*}{Polynomial ambiguity} & Letter-bounded & Polynomial ambiguity; \\  &  & CFL input; &  letter-bounded input; \\ & & Exponential ambiguity & commutative matrices \\  
\hline
\multirow{2}{*}{(Strict) Emptiness} & Undecidable \cite{Paz, FR17, DJ18} & \multirow{2}{*}{$\Longleftarrow$} & \multirow{2}{*}{Undecidable \cite{Bell19}}\\ & $\Longleftarrow$ & & \\
\hline
Cutpoint isolation & Undecidable \cite{BMT77, FR17} & \textbf{Decidable} & $\Longrightarrow$ \\
\hline
\end{tabu}
}
\end{center}
\caption{The decidability of problems under different restrictions on the PFA. The main result of this paper is shown in \textbf{boldface}. Symbol $\Longrightarrow$ denotes that decidability is implied by the decidability of the more general model; $\Longleftarrow$ denotes that undecidability is implied by the more restricted model.}\label{results1Tab}
\end{table}

The result is surprising since in order to solve the cutpoint isolation problem, we must solve two subproblems. Either the cutpoint $\lambda$ can be reached exactly (the $\lambda$-probability problem), or else it can only be approximated arbitrarily closely and is only reached exactly in some limit. As mentioned, the emptiness problem for cutpoint languages is undecidable for polynomially ambiguous PFA on letter-bounded languages, even when all matrices commute \cite{Bell19}. The proof of this result shows a construction of a PFA for which determining if a given $\lambda \in [0, 1]$ is ever reached (i.e., the $\lambda$-probability problem) is undecidable. This may at first seem to contradict the results of this paper, since the $\lambda$-probability problem is one of the two subproblems to be solved for cutpoint isolation. Why is there  no contradiction then? It comes from the fact that as the powers of matrices used in the PFA constructed in \cite{Bell19} increase, the PFA valuation tends towards the limit value $\lambda$. Therefore, this $\lambda$ is always non-isolated and hence the cutpoint isolation problem for such constructed PFA and $\lambda$ is decidable. However, determining if the PFA ever \emph{exactly} reaches $\lambda$ is undecidable. So, there is no contradiction with the results of this paper. Our main result is stated as follows.

\begin{thm}\label{mainthm}
The cutpoint isolation problem for probabilistic finite automata where inputs are constrained to a given letter-bounded context-free language is decidable. Moreover, if the cutpoint is isolated, then a separation bound $\epsilon > 0$ can be computed such that no input word's acceptance probability lies within $\epsilon$ of the cutpoint.
\end{thm}

The proof of Theorem~\ref{mainthm} is found in Section~\ref{algSec}. Our proof technique for showing the decidability of cutpoint isolation for PFA on letter-bounded languages uses the following crucial facts. If a PFA over a letter-bounded context-free language  can approach some given cutpoint $\lambda$ arbitrarily closely, then the PFA can reach $\lambda$ \emph{exactly} if we allow a subset of the matrices to be taken to one of their `limiting powers'. We use the property that each limiting power (of which there may be finitely many) of a stochastic matrix can be computed in polynomial time (see Lemma~\ref{limitsCompute}), as well as a crucial property from linear algebra that dominant eigenvalues (those of strictly largest magnitude) of a stochastic matrix are necessarily of magnitude $1$, roots of unity and they have equal geometric and algebraic multiplicities (see Lemma~\ref{rootsLem}). Since the input words of the PFA come from a letter-bounded CFL, we also use the fact that a letter-bounded language is context-free if and only if its Parikh image is a stratified semilinear set (see Proposition~\ref{parikhProp}). 

The combination of these ideas allows us to derive Algorithm~\ref{thealg}, which works as follows. We initially set all variables as free (rather than fixed), and compute the Parikh image $p(\mathcal{L})$ of the given letter-bounded CFL $\mathcal{L}$. Using the fact that $p(\mathcal{L})$ is a semilinear set, we compute which letters can be taken to arbitrarily high powers and which letters have fixed finite values. We then use the technical Proposition~\ref{limitlem} which states that if we can reach $\lambda$ then we can either do so by setting all free variables to an infinite power, or else we can compute an integer $C$ such that the value of one of free variables must be less than $C$. We then either set all free variables as $\omega$ in the first case, or nondeterministically choose one of the free variables and assign it a value less than $C$ in the latter case. In the second case we also update the semilinear set and repeat the above procedure until no free variables remain. Finally, we verify that the PFA has exactly the value $\lambda$ for the chosen values of the variables.

The crucial Proposition~\ref{limitlem} is somewhat technical, but relies on splitting a product of stochastic matrices into a summation involving dominant and subdominant eigenvalues (a subdominant eigenvalue being one with magnitude strictly less than $1$) and then applying the spectral decomposition or Jordan normal form of each stochastic matrix in order to derive the constant $C$ which bounds the value of one of free variables. 

Combining our proof technique with a result of Rabin \cite{Ra63}, we derive the following result.
\begin{customcor}{\ref{newCor1}}
The emptiness problem is decidable for probabilistic finite automata on letter-bounded context-free languages when the cutpoint is isolated.
\end{customcor}
The undecidability of the emptiness problem for PFA over letter-bounded inputs shown in \cite{Bell19} therefore only applies when the cutpoint is non-isolated.

The provided algorithm is nondeterministic in nature although we do not have an upper bound on its complexity. We can however provide the following lower bound via an adaptation of a proof technique from \cite{Bell19} which proved the NP-hardness of the injectivity problem for linearly ambiguous three-state probabilistic finite automata over letter-bounded languages, the proof of which is given in the appendix.

\begin{thm}\label{npthm}
Cutpoint isolation is NP-hard for $3$-state PFA on letter-bounded languages.
\end{thm}

Our procedure also allows us to answer some equivalent problems (in Section~\ref{extensions}). Given a PFA, $\lambda \in [0, 1]$ and a fixed number of alternations between probabilistic events $k \in \N$, determine if $\lambda$ is isolated (i.e. we may apply each probabilistic event an unbounded number of times but only alternate between different events a fixed number of times). We also prove the value-$1$ problem is decidable over letter-bounded context-free language inputs.

\section{Preliminaries}

{\bf Notation.\mbox{ }} We denote by $\mathbb{F}^{n \times n}$ the set of all $n \times n$ matrices over some field $\mathbb{F}$. We will primarily be interested in rational matrices. We define the spectrum (set of eigenvalues) of $A \in \R^{n \times n}$ as $\sigma(A) = \{\lambda_{1}, \ldots, \lambda_{n}\}$ arranged in monotonically nonincreasing order, i.e. $|\lambda_{i}| \geq |\lambda_{j}|$ for all $1 \leq i < j \leq n$ and we define $\hat{\sigma}(A) \subseteq \sigma(A)$ as the set of eigenvalues of $A$ of absolute value $1$. We call eigenvalues $\hat{\sigma}(A)$ \emph{dominant eigenvalues} and eigenvalues $\sigma(A) \setminus \hat{\sigma}(A)$ \emph{subdominant eigenvalues}. 

Given $A = (a_{ij}) \in\mathbb{F}^{m\times m}$ and $B\in\mathbb{F}^{n\times n},$ we define the direct sum $A\oplus B$ of $A$ and $B$ by: $
A\oplus B=
\left[\begin{array}{@{}c|l@{}}
A & \0_{m,n}\\
\hline
\0_{n,m} & B
\end{array}\right],
$
where $\0_{n,m}$ is the zero $n\times m$ matrix.

We use a nonstandard form of \emph{Dirac bra-ket notation} in several calculations, to simplify the notation in some complex formulae. If $u = (u_1, \ldots, u_n)^\top \in \C^n$ is a column vector, then we write $\ket{u} = u$ and $\bra{u} = u^\top$ where $u^\top$ denotes the transpose of $u$, i.e., $\bra{u} = (u_1, \ldots, u_n)$. Note that Dirac bra-ket notation ordinarily defines that $\bra{u} = u^*$ where $u^*$ denotes the \emph{conjugate} transpose of $u$, however we will not use this notion at any point. Note that $\ket{u}\bra{v}$ is just a rank 1 matrix $u^\top v$. We use $\bra{e_i}$ and $\ket{e_i}$ to denote the $i$'th basis row/column vector respectively. 

\noindent{\bf Probabilistic Finite Automata (PFA).\mbox{ }} A PFA $\mathcal{P}$ with $n$ states over an alphabet $\Sigma$ is defined as $\mathcal{A}=(\bra{u}, \{M_a|a \in \Sigma\}, \ket{v})$ where $\bra{u} \in \mathbb{R}^n$ is the initial probability distribution; $\ket{v} \in \{0, 1\}^n$ is the final state vector and each $M_a \in \mathbb{R}^{n \times n}$ is a (row) stochastic matrix. For a word $w = w_1w_2\cdots w_k \in \Sigma^*$, we define the acceptance probability $f_{\pfa}: \Sigma^* \to \mathbb{R}$ of $\pfa$ as:
\[
f_{\pfa}(w) = \bra{u} M_{w_1}M_{w_{2}} \cdots M_{w_k} \ket{v},
\]
which denotes the acceptance probability of $w$.\footnote{Some authors interchange the order of ${u}$ and ${v}$ and use column stochastic matrices, although the two definitions are trivially isomorphic.} 

For any $\lambda \in [0, 1]$ and PFA $\pfa$ over alphabet $\Sigma$, we define a cutpoint language to be: $L_{\geq \lambda}(\pfa) = \{w \in \Sigma^* | f_{\pfa}(w) \geq \lambda\}$, and a strict cutpoint language $L_{> \lambda}(\pfa)$ by replacing $\geq$ with $>$. The (strict) emptiness problem for a cutpoint language is to determine if $L_{\geq \lambda}(\pfa) = \emptyset$ (resp. $L_{> \lambda}(\pfa) = \emptyset$). Our main focus is on the \emph{cutpoint isolation problem}, now defined. 

\begin{prob}[Cutpoint isolation]
	Given a PFA $\mathcal{P}$ and cutpoint $\lambda \in [0, 1]$, determine if for each $\epsilon > 0$ there exists some $w \in \Sigma^*$ such that $|f_{\pfa}(w) - \lambda| < \epsilon$.
\end{prob}

Let $\Sigma = \{a_1, a_2, \ldots, a_\ell\}$ be an alphabet with $\ell>0$ distinct letters. A language $\mathcal{L}$ is called \emph{letter-bounded} if $\mathcal{L} \subseteq a_1^*a_2^*\cdots a_\ell^*$. If $\mathcal{L}$ is letter-bounded and also a  context-free language, then it is called a letter-bounded  context-free language. We are interested in cutpoint isolation for PFA whose inputs come from a given letter-bounded context-free language.

For a letter-bounded language $\mathcal{L} \subseteq a_1^*a_2^* \cdots a_\ell^*$, define its \emph{Parikh image}\footnote{In general, the Parikh image of $\mathcal{L} \subseteq\Sigma^*$ is defined as $p(\mathcal{L}) = \{\, (|w|_{a_1},\ldots,|w|_{a_\ell})\ :\ w\in \mathcal{L}\,\}$ where $|w|_{a_i}$ denotes the number of occurrences of letter $a_i$ in word $w$.} as
\[
p(\mathcal{L}) = \{\, (k_1,\ldots,k_\ell)\ :\ a_1^{k_1} a_2^{k_2} \cdots a_\ell^{k_\ell}\in \mathcal{L}\,\}.
\]
Recall that a subset $Q\subseteq \N^\ell$ is called \emph{linear} if there are vectors $q_0,q_1,\ldots,q_r\in \N^\ell$ such that
\[
Q = \{\,q_0 + t_1q_1 + \cdots +t_r q_r\ :\ t_1,\ldots,t_r\in\N\,\}.
\]
We say that a linear set $Q$ is \emph{stratified} if for each $i\geq 1$ the vector $q_i$ has at most two nonzero coordinates, and for any $i,j\geq 1$ if both $q_i$ and $q_j$ have two nonzero coordinates, $i_1<i_2$ and $j_1<j_2$, respectively, then their order is \emph{not} $i_1<j_1<i_2<j_2$, i.e., they are not interlaced.
A finite union of linear sets is called a \emph{semilinear set}, and a finite union of stratified linear sets is called a \emph{stratified semilinear set}.

We will need the following classical fact about context-free languages.
\begin{prop}\label{parikhProp}
	If $\mathcal{L}$ is a context-free language, then its Parikh image $p(\mathcal{L})$ is a semilinear set that can be effectively constructed from the definition of $\mathcal{L}$ \cite{Par66}.
\end{prop}

\begin{remark}
	There is a nice characterization of the letter-bounded context-free languages. Namely, a letter-bounded language $\mathcal{L} \subseteq a_1^*a_2^* \cdots a_\ell^*$ is context-free if and only if $p(\mathcal{L})$ is a stratified semilinear set \cite{Gins66,GS66}.
\end{remark}

\noindent{\bf Probabilistic Finite Automata on Letter-Bounded Languages.} Let $A_1, \ldots, A_\ell \in \mathbb{Q}^{n \times n}$ be row stochastic matrices. Let $u \in \mathbb{Q}^{n}$ be a stochastic vector (the initial vector) and $v \in \{0, 1\}^n$ (the final state vector).
Let $\mathcal{L} \subseteq a_1^*a_2^* \cdots a_\ell^*$ be a letter-bounded context-free language, and let $\lambda \in [0, 1]$ be a cutpoint for which we want to decide if it is isolated or not, that is, whether $\lambda$ belongs to the closure of $ \{\, \bra{u} A_1^{k_1} A_2^{k_2} \cdots A_\ell^{k_\ell} \ket{v}\ :\ a_1^{k_1} a_2^{k_2} \cdots a_\ell^{k_\ell}\in \mathcal{L}\,\}$.

If $\lambda$ is not isolated, then there are two scenarios: either there exists $k_1, k_2, \ldots, k_\ell \in\N$ such that $\bra{u} A_1^{k_1} A_2^{k_2} \cdots A_\ell^{k_\ell}\ket{v} = \lambda$,
or else $\lambda$ is never reached but only approached arbitrarily closely. In the second case there is a sequence of tuples $\{(k^m_1, k^m_2, \ldots, k^m_\ell)\}_{m=1}^\infty$ such that
\[
\lambda = \lim_{m\to\infty} \bra{u} A_1^{k^m_1} A_2^{k^m_2} \cdots A_\ell^{k^m_\ell}\ket{v}
\]
and, furthermore, for every $t\in\{1,\ldots,\ell\}$, either $k^m_t=k^1_t$ for all $m\geq 1$, i.e.\ $k^m_t$ is fixed, or $k^m_t$ is strictly increasing and $A_t^{k^m_t}$ converges to a limit as $m\to\infty$. It follows that if $\lambda$ is not isolated, then there exists a choice of variables $k_1, k_2, \ldots, k_\ell \in \N\cup \{\omega\}$ such that
\[
\lambda \in \bra{u} A_1^{k_1} A_2^{k_2} \cdots A_\ell^{k_\ell}\ket{v}.
\]
Note that if $k_t=\omega$, then $A_t^\omega$ may be a finite \emph{set} of limits (see Lemma \ref{limitsCompute} below for a detailed explanation).
In this case we substitute all limits of $A_t^\omega$ in the above formula, and so $\bra{u}  A_1^{k_1} A_2^{k_2} \cdots A_\ell^{k_\ell}\ket{v}$ also becomes a finite set\footnote{In the case if all $k_t$'s are finite or if all limits are unique, we identify the number $\bra{u} A_1^{k_1} A_2^{k_2} \cdots A_\ell^{k_\ell}\ket{v}$ with the one element set $\{\bra{u} A_1^{k_1} A_2^{k_2} \cdots A_\ell^{k_\ell}\ket{v}\}$.}.

\noindent{\bf Algebraic numbers.} A complex number $\alpha$ is \emph{algebraic} if it is a root of a polynomial $p\in\Z[x]$. The \emph{defining polynomial} $p_\alpha\in\Z[x]$ for $\alpha$ is the unique polynomial of least degree with positive leading coefficient such that the coefficients of $p_\alpha$ do not have a common factor and $p_\alpha(\alpha)=0$. The \emph{degree} and \emph{height} of $\alpha$ are defined to be that of $p_\alpha$.

In order to do computations with algebraic numbers we use their standard representations. Namely, an algebraic number can be represented by its defining polynomial and a sufficiently good complex rational approximation. More precisely, $\alpha$ will be represented by a tuple $(p_\alpha,a,b,r)$, where $p_\alpha\in\Z[x]$ is the defining polynomial for $\alpha$ and $a,b,r\in\Q$ are such that $\alpha$ is the unique root of $p_\alpha$ inside the circle in $\mathbb{C}$ with centre $a+bi$ and radius $r$. As shown in \cite{Mi76}, if $\alpha\neq \beta$ are roots of $p\in\Z[x]$, then $|\alpha-\beta| > \dfrac{\sqrt{6}}{d^{(d+1)/2}H^{d-1}}$, where $d$ and $H$ are the degree and height of $p$, respectively. So, if we require $r$ to be smaller than half of this bound, the above representation is well-defined.

Let $||\alpha||$ be the size of the standard representation of $\alpha$, that is, the total bit size of $a,b,r$ and the coefficients of $p_\alpha$. It is well-known fact that for given algebraic numbers $\alpha$ and $\beta$, one can compute $1/\alpha$, $\bar{\alpha}$ and $|\alpha|$  in time polynomial in $||\alpha||$, and one can compute $\alpha+\beta$ and $\alpha\beta$ and decide whether  $\alpha=\beta$ in time polynomial in $||\alpha|| + ||\beta||$. Moreover, for a \emph{real} algebraic $\alpha$, deciding whether $\alpha>0$ can be done in time polynomial in $||\alpha||$. Finally, there is a polynomial time algorithm that for a given $p\in\Z[x]$ computes the standard representations of all roots of~$p$. For more information on efficient algorithmic computations with algebraic numbers the reader is referred to \cite{BPR06,Cohen93,HHHK05,Pan96}.

\noindent {\bf Spectral decomposition and Jordan normal forms.} We will use both the \emph{spectral decomposition theorem} and the \emph{Jordan normal form} of stochastic matrices in later proofs. For background, see \cite{fried}.

Let $A_i \in \mathbb{Q}^{n \times n}$ be a matrix (we use notation $A_i$ since it will prove useful in the proof of Proposition~\ref{limitlem}), and let $\{\lambda_{i,1},\ldots,\lambda_{i,n_i}\}$ be the eigenvalues
of $A_i$ listed according to their \emph{geometric} multiplicities\footnote{Note that $n_i$ is the number of linearly independent eigenvectors of $A_i$ or the number of Jordan blocks in the Jordan normal form of $A_i$. The matrix $A_i$ is diagonalizable if and only if $n_i = n$. Jordan normal forms are unique up to permutations of the Jordan blocks.}. Then $A_i$ can be written in \emph{Jordan normal form} $A_i = S_i^{-1} (J_{\ell_{i,1}}(\lambda_{i, 1}) \oplus \cdots \oplus J_{\ell_{i,n_{i}}}(\lambda_{i, n_{i}})) S_i$, where $S_i$ is an invertible matrix (det$(S_i) \neq 0$) and $J_{\ell_{i,j}}(\lambda_{i, j})$ is a $\ell_{i,j} \times \ell_{i,j}$ \emph{Jordan block} for $1 \leq j \leq n_i \leq n$, with $n_{i}$ the number of Jordan blocks of $A_i$ and $\ell_{i,j}$ the size of the Jordan block corresponding to eigenvalue $\lambda_{i, j}$, such that $\ell_{i,1} + \cdots + \ell_{i,n_{i}} = n$. Jordan block $J_{\ell_{i,j}}(\lambda_{i, j})$ corresponds to the $j^{\textnormal{th}}$ eigenvalue $\lambda_{i, j}$ of $A_i$ and has the form:
\[
J_{\ell_{i,j}}(\lambda_{i, j}) = \begin{pmatrix}\lambda_{i, j} & 1 & 0 & \cdots & 0 \\ 0 & \lambda_{i, j} & 1 & \cdots & 0 \\ 0 & 0 & \lambda_{i, j} & \cdots & 0 \\ \vdots & \vdots & \vdots & \ddots & \vdots \\ 0 & 0 & 0 & \cdots & \lambda_{i, j} \end{pmatrix} \in \mathbb{C}^{\ell_{i,j} \times \ell_{i,j}}
\]
The matrix $S_i$ contains the \emph{generalised} eigenvectors of $A_i$. Noting that $\binom{x}{y} = 0$ if $y > x$, we now see that
\begin{eqnarray}
& & J_{\ell_{i,j}}(\lambda_{i, j})^{k_i} = \begin{pmatrix}\lambda_{i, j}^{k_i} & \binom{k_i}{1}\lambda_{i, j}^{k_i-1} & \binom{k_i}{2}\lambda_{i, j}^{k_i-2} & \cdots & \binom{k_i}{\ell_{i,j}-1}\lambda_{i, j}^{k_i-(\ell_{i,j}-1)} \\ 0 & \lambda_{i, j}^{k_i} & \binom{k_i}{1}\lambda_{i, j}^{k_i-1} & \cdots & \binom{k_i}{\ell_{i,j}-2}\lambda_{i, j}^{k_i-(\ell_{i,j}-2)} \\ 0 & 0 & \lambda_{i, j}^{k_i} & \cdots & \binom{k_i}{\ell_{i,j}-3}\lambda_{i, j}^{k_i-(\ell_{i,j}-3)} \\ \vdots & \vdots & \vdots & \ddots & \vdots \\ 0 & 0 & 0 & \cdots & \lambda_{i, j}^{k_i} \end{pmatrix} \in \mathbb{C}^{\ell_{i,j}\times \ell_{i,j}} \nonumber \\
& = & \sum_{0 \leq m \leq \ell_{i,j}-1} \lambda_{i, j}^{k_i-m} \binom{k_i}{m}\left( \sum_{1 \leq p \leq \ell_{i,j}-m} \ket{e_p}\bra{e_{m+p}}\right) \label{jordanform}
\end{eqnarray}

The \emph{spectral decomposition} of a matrix is a special case of the Jordan normal form. Namely, any \emph{diagonalizable} matrix $A_i \in \mathbb{Q}^{n \times n}$ can be written as
\begin{equation}\label{spectralform}
A_i = S_i^{-1} (\lambda_{i, 1} \oplus \cdots \oplus \lambda_{i, n}) S_i = \sum_{j = 1}^{n} \lambda_{i,j} \ket{v_{i,j}}\bra{u_{i,j}},
\end{equation}
where $\sigma(M) = \{\lambda_{i,1}, \ldots, \lambda_{i,n}\}$ is the set of eigenvalues of $A_i$, $\ket{v_{i,j}}$ is the $j$'th column of $S_i^{-1}$ and $\bra{u_{i,j}}$ is the $j$'th row of $S_i$. Thus we have $A_i^k = \sum_{j = 1}^{n} \lambda^k_{i,j} \ket{v_{i,j}}\bra{u_{i,j}}$.

We will also require the following technical lemma concerning the dominant eigenvalues of stochastic matrices.

\begin{lem}[{\cite[Theorem 6.5.3]{fried}}]\label{rootsLem}
Let $\lambda$ be a dominant eigenvalue of a stochastic matrix $A \in \R^{n \times n}$. Then $\lambda$ is a root of unity of order no more than $n$. Moreover, the geometric multiplicity of $\lambda$ is equal to its algebraic multiplicity. In other words, the Jordan blocks that correspond to $\lambda$ have size $1\times 1$.
\end{lem}

We also require the following lemma. 

\begin{lem}\label{limitsCompute}
For any stochastic matrix $A$, the sequence $\{A^k\}_{k=1}^\infty$ has a finite number of limits. Namely, there exist a \emph{computable} constant $d$ such that, for each $r=0,\ldots,d-1$, the subsequence $\{A^{dm+r}\}_{m=1}^\infty$ converges to a limit, and this limit can be computed in polynomial time given $d$ and $r$.
\end{lem}

\begin{proof}
Let $A$ be a stochastic matrix. As shown in \cite{Cai94}, we can compute in polynomial time the Jordan normal form of $A$ and a transformation matrix $S$ such that $A=S^{-1}JS$. Note that $A$ may have complex eigenvalues, so all computations are done using standard representations of algebraic numbers.

By Lemma \ref{rootsLem}, all dominant eigenvalues of $A$ are roots of unity of orders no more than $n$, and their Jordan blocks have size $1\times 1$. If $\lambda$ is a root of unity of order $p$, then $\{\lambda^k\}_{k=1}^\infty$ is a periodic sequence with period $p$. On the other hand, if $J_\ell(\lambda)$ is a Jordan block corresponding to an eigenvalue $\lambda$ such that $|\lambda|<1$, then $\lim_{k\to\infty} J_\ell(\lambda)^k$  is equal to the zero matrix.

Let $d$ be the least common multiple of the orders of the roots of unity among the eigenvalues of $A$. Now if $\lambda$ is a dominant eigenvalue of $A$, then the values of $\lambda^{dm+r} = \lambda^r$ do not depend on $m$, where $r=0,\ldots,d-1$. Hence $J^{dm+r}$ converges to a limit when $m\to\infty$. This limit is equal to a matrix $J'$ obtained from $J$ by replacing all dominant $\lambda$ with $\lambda^r$ and all Jordan blocks corresponding to subdominant eigenvalues with zero matrices. So, $\lim_{m\to\infty} A^{dm+r} = S^{-1}J'S$.

This shows that $\{A^k\}_{k=1}^\infty$ has at most $d$ limits. Finally, we note that $d$ may be exponential in the dimension of $A$. However, if $\{A^k\}_{k=1}^\infty$ has a single limit, then this limit can be computed in polynomial time.
\end{proof}

For a stochastic matrix $A$, we will use notation $A^\omega$ to denote the set of all limits of the sequence $\{A^k\}_{k=1}^\infty$. If $\{A^k\}_{k=1}^\infty$ has a single limit $A'$, then we will identify the set $A^\omega =\{A'\}$ with the matrix $A'$ and write $A^\omega=A'$.

\section{Decidability of Cutpoint Isolation}\label{algSec}

In this section we will give a proof of Theorem \ref{mainthm} which is our main result. The crucial ingredient of our proof is the following technical proposition which will be proven in Section~\ref{limitLemSec}.

\begin{prop}\label{limitlem}
Let $J = \{1, 2, \ldots, \ell\}$ be indices, $\lambda\in [0,1]$ a cutpoint, and let $J_F \subseteq J$ be such that $k_t$ is a free variable, for $t \in J_F$, and $k_t$ is assigned a fixed finite value, for $t \in J\setminus J_F$.
Then
\begin{itemize}
	\item either $\lambda \in \bra{u} A_1^{k_1} A_2^{k_2} \cdots A_\ell^{k_\ell}\ket{v}$, where $k_t=\omega$ for all $t\in J_F$,
	\item or else there exists a constant $C>0$ such that $\lambda\in \bra{u} A_1^{k_1} A_2^{k_2} \cdots A_\ell^{k_\ell}\ket{v}$ implies $k_t < C$ for at least one $t \in J_F$.
\end{itemize}
Moreover, we can decide whether the first case holds and compute the constant $C$ in the second case.
\end{prop}

Below we give a high-level description of the main algorithm (Algorithm~\ref{mainalg}), which gives a formal proof of Theorem~\ref{mainthm}, and explain how Proposition~\ref{limitlem} is used there.

Let $\mathcal{L} \subseteq a_1^*a_2^* \cdots a_\ell^*$ be a given letter-bounded CFL. We start by considering all indices $J=\{1,\ldots,\ell\}$ as \emph{free} (i.e. their value is not fixed and they will later be given a fixed value from $\N \cup \{\omega\}$) and iteratively fix them until no free indices remain. We first use Parikh and Ginsburg's results (Proposition~\ref{parikhProp}) to compute the Parikh image $p(\mathcal{L})$. Then we nondeterministically choose a linear subset $Q$ and use it to determine the indices which can be taken to arbitrary high values while staying within $Q$. These indices will correspond to the ``free variables'' in the algorithm.

Let $J_F$ be a set of such indices (which will be called $R$ in Algorithm~\ref{mainalg}). We then set $k_t$ for $t \in J \setminus J_F$ to appropriate finite values, while $k_t$ with $t \in J_F$ remain free variables. We wish to determine if there is a choice of $k_t \in \N \cup \{\omega\}$ for $t \in J_F$ such that
\[
\lambda \in \bra{u} A_1^{k_1} A_2^{k_2} \cdots A_\ell^{k_\ell}\ket{v},
\]
that is, whether $\lambda$ can be reached by setting each free variable $k_t$ either to some finite value or else to $\omega$, an ``infinite'' power.  Proposition~\ref{limitlem} then tells us that either all free variables should be set at $\omega$ in order to reach $\lambda$ (and this is decidable), or else there exists a computable constant $C$ such that \emph{if} we can reach $\lambda$ by some choice of these free variables, then some $k_t < C$ for an index $t \in J_F$.

In the first case, we set all free variable to $\omega$. In the second case, we nondeterministically choose some free variable, fix its value in the range $[0, C)$ and then update our linear set $Q$ to satisfy a new constraint. The procedure repeats iteratively until all free variables have been assigned a fixed value. The algorithm then verifies if this choice of variables gives a solution.

\hrulefill

\begin{algorithm}[H]\label{mainalg}
\caption{Nondeterministic algorithm deciding whether a given cutpoint is isolated.} \label{thealg}

\begin{algorithmic}
    \STATE \textbf{Stage one} (Nondeterministic iterative fixing of free variables):
    \bindent
	\STATE{Let $T = J = \{1, 2, \ldots, \ell\}$.}
	\STATE{Compute the Parikh image $p(\mathcal{L})$ and nondeterministically choose one of its finitely}
	\STATE{many linear subsets $Q = \{\,q_0 + t_1q_1 + \cdots +t_r q_r\ :\ t_1,\ldots,t_r\in\N\,\} \subseteq p(\mathcal{L})$.\footnote{Here we use the fact that if $\lambda$ can be approached arbitrarily closely within  $p(\mathcal{L})$, then $\lambda$ can be approached by staying within one of the finitely many linear subsets of $p(\mathcal{L})$.}}
	\WHILE {$T \neq \emptyset$}
		\STATE{Let $R$ be the set of indices $j\in T$ such that at least one $q_i$ with $i\geq 1$ has a nonzero $j$th coordinate.\footnote{Thus $R$ is the subset of indices from $T$ that can be taken to arbitrarily large powers simultaneously.}}
		\STATE{For each $j\in T\setminus R$, the $j$th coordinate of all vectors from $Q$ is equal to the $j$th coordinate of $q_0$.\footnote{This is because for $j\in T\setminus R$ all $q_i$ with $i\geq 1$ have zero $j$th coordinate.} So, we set $k_j$ for $j\in T\setminus R$ to be the $j$th coordinate of $q_0$.}
			\STATE{Then, for $j\in R$, compute the limits $A_j^\omega$ with indices respecting set $Q$\\ (see Remark~\ref{rem:1} below for details).}
			\STATE{Check whether $\lambda \in \bra{u} A_1^{k_1} A_2^{k_2} \cdots A_\ell^{k_\ell}\ket{v}$, where $k_j=\omega$ for all $j\in R$.}
			\STATE{If yes, return True and stop.}
			\STATE{Otherwise, assuming all indices $k_j$ for $j \in J\setminus R$ are fixed and $R$ is the set of free variables, use Proposition~\ref{limitlem} to compute the constant $C > 0$ such that $\lambda\in \bra{u} A_1^{k_1} A_2^{k_2} \cdots A_\ell^{k_\ell}\ket{v}$ implies $k_j < C$ for at least one $j \in R$.}
			\STATE{Then nondeterministically choose $j \in R$, fix $k_j \in [0, C)$ and set $T \leftarrow R \setminus \{j\}$.}
			\STATE{Next, for the chosen index $j$, find those indices $i$ in $\{1,\ldots,r\}$ for which $q_i$ has a nonzero $j$th coordinate. Without loss of generality, suppose $\{1, \ldots, s\}$ are these indices.}
			\STATE{Fixing $k_j\in [0, C)$, restricts the parameters $t_1, \ldots, t_s$ in $Q$ to a finite set of possible values since the vector $q_0 + t_1q_1 + \cdots +t_s q_s$ must have $k_j$ in its $j$th coordinate.}
			\STATE{Nondeterministically choose one of these values for $t_1,\ldots,t_s$ or return False and stop, if such a choice is impossible.}
			\STATE{Let $Q \leftarrow \{\,(q_0 + t_1q_1 + \cdots +t_s q_s) + t_{s+1} q_{s+1} + \cdots +t_r q_r\ :\ t_{s+1},\ldots,t_r\in\N\,\}$.\footnote{Thus $(q_0 + t_1q_1 + \cdots +t_s q_s)$ becomes the new value of $q_0$ in $Q$.}}
	\ENDWHILE
    \eindent
\STATE \textbf{Stage two} (Verifying the computation):
    \bindent
    \STATE{At this stage we have fixed all variables $k_1, k_2, \ldots,k_\ell$ to some finite values.}
	\STATE{Compute $\bra{u} A_1^{k_1} A_2^{k_2} \cdots A_{\ell}^{k_\ell} \ket{v}$ for the obtained values of $k_1,\ldots,k_\ell\in \N$.} 
	\STATE{Return True if $\lambda = \bra{u} A_1^{k_1} A_2^{k_2} \cdots A_\ell^{k_\ell}\ket{v}$ or False, otherwise.}
    \eindent
\STATE \textbf{End.}
\vspace{0.2cm}
\end{algorithmic}
\end{algorithm}

\hrulefill

\begin{remark}\label{rem:1}
	To compute the limits $A_j^\omega$ with indices respecting set $Q$, note that the projection of $Q$ on the $j$th coordinate is equal to $\{\,q_{0,j} + t_1q_{1,j} + \cdots +t_r q_{r,j}\ :\ t_1,\ldots,t_r\in\N\,\}$, where $q_{i,j}$ is the $j$th coordinate of $q_i$. The set $\langle q_{1,j},\ldots,q_{r,j}\rangle = \{t_1q_{1,j} + \cdots +t_r q_{r,j}\ :\ t_1,\ldots,t_r\in\N\,\}$ is a finitely generated subsemigroups of $(\N,+)$. Let $d=\gcd( q_{1,j},\ldots,q_{r,j})$, then there is a number $s>0$ such that for any $t\geq s$, we have $t\in \langle q_{1,j},\ldots,q_{r,j}\rangle$ if and only if $d$ divides $t$. This is a well-known property of the subsemigroups of $(\N,+)$ \cite{RS09}.
	Thus the limits $A_j^\omega$ with indices respecting $Q$ are equal to $A_j^\omega = A_j^{q_{0,j}}(A_j^d)^\omega$, where the limits $(A_j^d)^\omega$ are computed using Lemma \ref{limitsCompute}.
\end{remark}

It remains to prove that we can compute a separation bound $\epsilon > 0$ between $\lambda$ and the closest acceptance probability of $\pfa$ for any input word $w \in \mathcal{L}$. Algorithm~\ref{mainalg} has two stopping conditions, either by returning True in \textbf{Stage one} (which we discount since it implies the cutpoint is not isolated), or else after \textbf{Stage two}.

Algorithm~\ref{mainalg} has two sources of nondeterminism in \textbf{Stage~one}: in the choice of linear subset $Q$ and then during the while loop in the choice of $j \in R$ and $k_j \in [0, C)$. We will evaluate every choice of linear subset $Q$ and every choice of $j$ and $k_j$ to cover all possible cases, updating a global variable $\epsilon$ at the end of every nondeterministic branch. Initially, we set $\epsilon \leftarrow \infty$, and let $\epsilon_1$, $\epsilon_2$ be additional global variables that are set $\epsilon_1 \leftarrow \epsilon_2 \leftarrow \infty$ at the beginning of every nondeterministic branch.

During the execution of \textbf{Stage one} we use Proposition~\ref{limitlem} to compute $C$ such that if all free variables are above $C$ then we are at least some $\epsilon' > 0$ away from $\lambda$. Note that $\epsilon'$ is less than half of the distance between $\lambda$ and some limit values. For each iteration of the while loop, we set $\epsilon_1 \leftarrow \min\{\epsilon_1, \epsilon'\}$ to keep track of the minimal value. During \textbf{Stage two}, all variables have a fixed finite value, and we set $\epsilon_2 \leftarrow |\bra{u} A_1^{k_1} A_2^{k_2} \cdots A_\ell^{k_\ell}\ket{v} - \lambda|$ which is greater than zero assuming $\lambda$ is isolated. Finally, we set $\epsilon \leftarrow \min\{\epsilon, \epsilon_1, \epsilon_2\} > 0$.

After inspecting all possible nondeterministic runs of the algorithm, the obtained value of $\epsilon$ gives us the separation bound. Indeed, during the execution of the above procedure, $\epsilon$ is updated to the minimum of $\epsilon_1$ and $\epsilon_2$, where $\epsilon_1$ is less than half of the distance between $\lambda$ and some limit values and $\epsilon_2$ keeps track of the distance between $\lambda$ and the values $\bra{u} A_1^{k_1} A_2^{k_2} \cdots A_\ell^{k_\ell}\ket{v}$ when each index $k_j$ is less than the corresponding constant $C$.

\section{Proof of Proposition~\ref{limitlem}}\label{limitLemSec}

We begin with a proof sketch. Since each $A_i$ is stochastic, $\hat{\sigma}(A_i)$ contains at least one eigenvalue $1$ and all other eigenvalues in $\hat{\sigma}(A_i)$ are roots of unity by Lemma~\ref{rootsLem}. All eigenvalues in $\sigma(A_i) \setminus \hat{\sigma}(A_i)$ have absolute value strictly smaller than $1$. Our approach is to rewrite the expression
\begin{equation}\label{basicForm}
	\bra{u} A_1^{k_1} A_2^{k_2} \cdots A_\ell^{k_\ell}\ket{v}
\end{equation}
into the sum of two terms (which will be denoted by $S_0$ and $S_1$) such that $S_0$ determines the limit behaviour as all free variables tend towards infinity, since they control only dominant eigenvalues, while $S_1$ is vanishing, since at least one free variable controls a subdominant eigenvalue. We can then reason that if all free variables simultaneously become larger, then Eqn~(\ref{basicForm}) tends towards a set of computable limits with some vanishing terms. Therefore we can determine either that we can reach $\lambda$ when all free variables are $\omega$, or else we can prove that Eqn~(\ref{basicForm}) is within any $\epsilon>0$ of a limit value once all free variables are sufficiently large, which proves the proposition (by setting $\epsilon$ as less than the smallest difference from a limit value and $\lambda$). We now proceed with the formal details. 

First, we consider the simpler case when all matrices are diagonalizable and then show how to extend this argument to the general case.

\noindent{\bf Diagonalizable matrices.} Let us first assume that all matrices are diagonalizable. By the spectral decomposition theorem (see Eqn~(\ref{spectralform})), we may write a matrix $A_i^{k_i}$ as:
\begin{eqnarray}
A_i^{k_i} & = & \sum_{j=1}^n\lambda_{i,j}^{k_i} \ket{v_{i,j}}\bra{u_{i,j}},
\end{eqnarray}
where $\{\lambda_{i,1},\ldots,\lambda_{i,n}\}$ are the eigenvalues of $A_i$ repeated according to their multiplicities, and the vectors $\ket{v_{i,j}}$ and $\bra{u_{i,j}}$, for $1 \leq j \leq n$, are related to the eigenvectors of $A_i$. Now, we can write:
\begin{eqnarray*}
& & \bra{u} A_1^{k_1} A_2^{k_2} \cdots A_\ell^{k_\ell} \ket{v}
  =  \bra{u} \left( \prod_{i = 1}^{\ell} \left( \sum_{j=1}^n\lambda_{i,j}^{k_i} \ket{v_{i,j}}\bra{u_{i,j}}\right) \right) \ket{v} \\
 & = & \sum_{j_1, \ldots, j_\ell \in [1,n]} \lambda_{1,j_1}^{k_1} \lambda_{2,j_2}^{k_2} \cdots \lambda_{\ell,j_\ell}^{k_\ell} \braket{u|v_{1,j_1}}\braket{u_{1,j_1}|v_{2,j_2}} \bra{u_{2,j_2}} \cdots \ket{v_{\ell,j_\ell}} \braket{u_{\ell,j_\ell}|v}
\end{eqnarray*}
Let us thus define $\Theta_{j_1, \ldots, j_\ell} = \braket{u|v_{1,j_1}}\braket{u_{1,j_1}|v_{2,j_2}} \bra{u_{2,j_2}} \cdots \ket{v_{\ell,j_\ell}} \braket{u_{\ell,j_\ell}|v}$. The above sum can be split in two: the first summand containing terms where only dominant eigenvalues are to the power of free variables, and the second containing terms with at least one subdominant eigenvalue to the power of a free variable (these two terms are labelled $S_0$ and $S_1$ below). This is a useful decomposition since any term which contains a subdominant eigenvalue taken to the power of a free variable will tend towards zero as the values of all free variables (simultaneously) increase. Thus we can write $\bra{u} A_1^{k_1} A_2^{k_2} \cdots A_\ell^{k_\ell} \ket{v}=S_0+S_1$, where
\begin{align*}
S_0 &= \sum_{\substack{j_1, \ldots, j_\ell \in [1, n] \\ \forall t \in J_F\ :\ |\lambda_{t, j_t}| = 1}} \lambda_{1,j_1}^{k_1} \lambda_{2,j_2}^{k_2} \cdots \lambda_{\ell,j_\ell}^{k_\ell} \Theta_{j_1, \ldots, j_\ell},\\
S_1 &= \sum_{\substack{j_1, \ldots, j_\ell \in [1, n] \\ \exists  t \in J_F\ :\ |\lambda_{t, j_t}| < 1}} \lambda_{1,j_1}^{k_1} \lambda_{2,j_2}^{k_2} \cdots \lambda_{\ell,j_\ell}^{k_\ell} \Theta_{j_1, \ldots, j_\ell}.
\end{align*}
By Lemma~\ref{rootsLem} the dominant eigenvalues are roots of unity, and so $S_0$ assumes only finitely many different values as $k_t$ with $t\in J_F$ vary, while $k_t$ with $t\in J\setminus J_F$ are fixed.

Suppose $S_1$ in not an empty sum since otherwise $S_1=0$. Then there exists $t \in J_F$ and $j_t\in [1,n]$ such that $|\lambda_{t, j_t}| < 1$. Let $\rho$ be the maximum among such values, that is,
\[
\rho = \max \{\,|\lambda_{t, j}|\ :\ t \in J_F,\ j\in [1,n] \text{ and } |\lambda_{t, j}| < 1\,\}.
\]
Suppose $k_t\geq C$ for $t\in J_F$, where $C$ is some constant to be chosen later. Then $S_1$ can be estimated as follows: since for every choice of $j_1, \ldots, j_\ell$ in the summation $S_1$ there is $t\in J_F$ with $|\lambda_{t, j_t}| \leq \rho < 1$ and $|\lambda_{i, j}| \leq 1$ for all other $\lambda_{i, j}$, we have
\[
|S_1|\leq C_1\rho^C,\quad \text{where}\quad C_1= \sum_{\substack{j_1, \ldots, j_\ell \in [1, n] \\ \exists  t \in J_F\ :\ |\lambda_{t, j_t}| < 1}} |\Theta_{j_1, \ldots, j_\ell}|.
\]
Notice that for any rational $\epsilon > 0$, we can compute $C \in \N$ such that $|S_1|\leq C_1\rho^{C} < \epsilon$.
Now, $S_0$ gives a finite number of limit values for $\bra{u} A_1^{k_1} A_2^{k_2} \cdots A_\ell^{k_\ell} \ket{v}$. If $\lambda$ is not equal to any of them, then choose $\epsilon>0$ to be less than half the minimal distance between $\lambda$ and those limit values. Using this $\epsilon$, we compute $C$ as above. By definition of $C$, if all $k_t\geq C$ for $t\in J_F$, then the distance between $\bra{u} A_1^{k_1} A_2^{k_2} \cdots A_\ell^{k_\ell} \ket{v}$ and one of the limit values of $S_0$ is less than~$\epsilon$. Thus $\bra{u} A_1^{k_1} A_2^{k_2} \cdots A_\ell^{k_\ell} \ket{v}$ cannot be equal to $\lambda$ when all $k_t\geq C$ for $t\in J_F$. Hence if $\lambda=\bra{u} A_1^{k_1} A_2^{k_2} \cdots A_\ell^{k_\ell} \ket{v}$, then there is $t\in J_F$ such that $k_t<C$.

\noindent{\bf The general case.} We now show how to extend the proof to the case when some matrices are non-diagonalizable.

Let $a_{i,j} = \sum_{s = 1}^{j-1}\ell_{i,s}$ be the sum of the sizes of the first $j-1$ Jordan blocks of matrix $A_i$, so that $a_{i,1} = 0, a_{i,2} = \ell_{i,1}, a_{i,3} = \ell_{i,1} + \ell_{i,2}$ etc. Then we see that by using Eqn~(\ref{jordanform}), $A_i^{k_i} = S_i^{-1} (J_{\ell_{i,1}}(\lambda_{i, 1})^{k_i} \oplus \cdots \oplus J_{\ell_{i,n_i}}(\lambda_{i, n_i})^{k_i}) S_i$ has the form
\begin{eqnarray*}
&  & S_i^{-1} \left( \sum_{1 \leq j \leq n_i} \sum_{0 \leq m \leq \ell_{i,j}-1} \lambda_{i,j}^{k_i-m} \binom{k_i}{m}\left( \sum_{1 \leq p \leq \ell_{i,j}-m} \ket{e_{a_{i,j}+p}}\bra{e_{a_{i,j} + m + p}}\right) \right) S_i \\
& = & \sum_{1 \leq j \leq n_i} \sum_{0 \leq m \leq \ell_{i,j}-1} \lambda_{i,j}^{k_i-m} \binom{k_i}{m}\left( \sum_{1 \leq p \leq \ell_{i,j}-m} S_i^{-1}\ket{e_{a_{i,j}+p}}\bra{e_{a_{i,j} + m + p}}S_i\right) \\
& = & \sum_{1 \leq j \leq n_i} \sum_{0 \leq m \leq \ell_{i,j}-1} \lambda_{i,j}^{k_i-m} \binom{k_i}{m}\left( \sum_{1 \leq p \leq \ell_{i,j}-m} \ket{v_{i,a_{i,j}+p}} \bra{u_{i,a_{i,j} + m + p}}\right),
\end{eqnarray*}
where $S_i = \sum_{q = 1}^{n} \ket{e_q}\bra{u_{i,q}}$ and $S_i^{-1} = \sum_{q = 1}^{n} \ket{v_{i,q}}\bra{e_q}$, with $e_q$ the $q^{\textnormal{th}}$ basis vector. Here we used the property that $\braket{e_i | e_j} = 0$ for any $i \neq j$. We may now compute that:
\begin{eqnarray*}
& & \bra{u} A_1^{k_1} A_2^{k_2} \cdots A_\ell^{k_\ell} \ket{v} \\
& = & \bra{u} \prod_{i = 1}^{\ell} \left(\sum_{1 \leq j \leq n_i} \sum_{0 \leq m \leq \ell_{i,j}-1} \lambda_{i,j}^{k_i-m} \binom{k_i}{m}\left( \sum_{1 \leq p \leq \ell_{i,j}-m} \ket{v_{i,a_{i,j}+p}} \bra{u_{i,a_{i,j} + m + p}}\right)  \right) \ket{v} \\
& = & \bra{u} \prod_{i = 1}^{\ell} \left(\sum_{1 \leq j \leq n_i} \sum_{0 \leq m \leq \ell_{i,j}-1} \lambda_{i,j}^{k_i-m} \binom{k_i}{m} \Psi_{i,j,m} \right) \ket{v}\\
& = & \sum_{\substack{j_1, \ldots, j_\ell\ |\  j_q \in [1, n_q]\\ m_1, \ldots, m_\ell\ |\  m_q \in [0, \ell_{q, j_q}-1]}} \left(\prod_{1 \leq t \leq \ell} \lambda_{t,j_t}^{k_t-m_t} \binom{k_t}{m_t} \right) \bra{u}\Psi_{j_1, \ldots, j_\ell}^{m_1, \ldots, m_\ell}\ket{v},
\end{eqnarray*}
where $\Psi_{i,j,m} = \sum\limits_{1 \leq p \leq \ell_{i,j}-m} \ket{v_{i,a_{i,j}+p}} \bra{u_{i,a_{i,j} + m + p}}$ and
\[
\Psi_{j_1, \ldots, j_\ell}^{m_1, \ldots, m_\ell} = \Psi_{1,j_1,m_1}\Psi_{2,j_2,m_2}\cdots \Psi_{\ell,j_\ell,m_\ell}.
\]
We may split the above summation, as before, into two parts corresponding to products containing only dominant eigenvalues to powers of free variables and those containing at least one subdominant eigenvalue to the power of a free variables. 
By Lemma~\ref{rootsLem}, Jordan blocks corresponding to dominant eigenvalues have size $1\times 1$, that is, if $|\lambda_{t,j_t}| = 1$ for $t\in J_F$, then $\ell_{t,j_t}=1$ and hence $m_t=0$. So, we can write
\[
\bra{u} A_1^{k_1} A_2^{k_2} \cdots A_\ell^{k_\ell} \ket{v} = S_0 + S_1,
\]
where
\begin{align}\label{defectiveFinalForm}
&S_0 =  \sum_{\substack{j_1, \ldots, j_\ell\ |\  j_q \in [1, n_q]\\ m_1, \ldots, m_\ell\ |\  m_q \in [0, \ell_{q, j_q}-1] \\ \forall t \in J_F\ :\ |\lambda_{t, j_t}| = 1}} \Bigg(\prod_{t\in J_F} \lambda_{t,j_t}^{k_t}\Bigg) \Theta_{j_1, \ldots, j_\ell}^{m_1, \ldots, m_\ell},\nonumber \\
&S_1 = \sum_{\substack{j_1, \ldots, j_\ell\ |\  j_q \in [1, n_q]\\ m_1, \ldots, m_\ell\ |\  m_q \in [0, \ell_{q, j_q}-1] \\ \exists t \in J_F\ :\ |\lambda_{t, j_t}| < 1}} \Bigg(\prod_{t\in J_F} \lambda_{t,j_t}^{k_t-m_t} \binom{k_t}{m_t} \Bigg) \Theta_{j_1, \ldots, j_\ell}^{m_1, \ldots, m_\ell}\qquad \text{and}\\
&\Theta_{j_1, \ldots, j_\ell}^{m_1, \ldots, m_\ell} = \Bigg(\prod_{t\in J\setminus J_F} \lambda_{t,j_t}^{k_t-m_t} \binom{k_t}{m_t} \Bigg) \bra{u}\Psi_{j_1, \ldots, j_\ell}^{m_1, \ldots, m_\ell}\ket{v}.\nonumber
\end{align}
Note that $k_t$'s in the formula for $\Theta_{j_1, \ldots, j_\ell}^{m_1, \ldots, m_\ell}$ are fixed since $t\notin J_F$ and
so $A_t$ is not a free matrix. In other words, $\Theta_{j_1, \ldots, j_\ell}^{m_1, \ldots, m_\ell}$ does not depend on free variables $k_i$ for $t\in J_F$. This also implies that $S_0$ assumes only finitely many different values as $k_t$ with $t\in J_F$ vary since by Lemma~\ref{rootsLem} the dominant eigenvalues are roots of unity.

Again using the fact that Jordan blocks corresponding to the dominant eigenvalues have size $1$$\times 1$,
we can rewrite the product inside the formula for $S_1$ from Eqn~(\ref{defectiveFinalForm}) as follows
\[
\prod_{t\in J_F} \lambda_{t,j_t}^{k_t-m_t} \binom{k_t}{m_t} =
\prod_{\substack{t\in J_F\\ |\lambda_{t,j_t}|=1}} \lambda_{t,j_t}^{k_t}\ \cdot
\prod_{\substack{t\in J_F\\ |\lambda_{t,j_t}|<1}} \lambda_{t,j_t}^{k_t-m_t} \binom{k_t}{m_t}.
\]

Suppose $S_1$ in not an empty sum since otherwise $S_1=0$. Then there exists $t \in J_F$ and $j_t\in [1,n_t]$ such that $|\lambda_{t, j_t}| < 1$. Let $\rho$ be the maximum among such values, that is,
\[
\rho = \max \{\,|\lambda_{t, j}|\ :\ t \in J_F,\ j\in [1,n_t] \text{ and } |\lambda_{t, j}| < 1\,\}.
\]
Notice that every summand in $S_1$ has at least one $|\lambda_{t,j_t}|\leq \rho<1$ with $t\in J_F$, and $|\lambda_{i,j}|\leq 1$ for all other $\lambda_{i,j}$. Also, $\binom{k_t}{m_t} \leq k_t^{m_t} \leq k_t^n$ since $m_t \leq n$. So every summand in $S_1$ can be estimated by the expression
\[
C_1\cdot\!\!\! \prod_{\substack{t\in J_F\\ |\lambda_{t,j_t}|<1}} \rho^{k_t} k_t^n, \qquad \text{where } C_1 \text{ is a computable constant}.
\]

We have $\rho^{k} k^n\to 0$ when $k\to \infty$, and for any rational $\delta>0$ we can compute $C$ such that $\rho^{k} k^n<\delta$ for $k\geq C$. If in addition we assume that $0<\delta<1$ and that $k_t\geq C$ for all $t\in J_F$, then $|S_1|\leq C_1 n^{2\ell}\delta$.

Now, $S_0$ gives a finite number of limit values for $\bra{u} A_1^{k_1} A_2^{k_2} \cdots A_\ell^{k_\ell} \ket{v}$. If $\lambda$ is not equal to any of them, then choose a rational $0<\delta<1$ such that $\epsilon = C_1 n^{2\ell}\delta$ is less than half the minimal distance between $\lambda$ and those limit values. Using this $\delta$, we compute $C$ as before. By definition of $C$, if all $k_t\geq C$ for $t\in J_F$, then the distance between $\bra{u} A_1^{k_1} A_2^{k_2} \cdots A_\ell^{k_\ell} \ket{v}$ and one of the limit values of $S_0$ is less than $\epsilon$. Thus $\bra{u} A_1^{k_1} A_2^{k_2} \cdots A_\ell^{k_\ell} \ket{v}$ cannot be equal to $\lambda$ when all $k_t\geq C$ for $t\in J_F$. Hence if $\lambda=\bra{u} A_1^{k_1} A_2^{k_2} \cdots A_\ell^{k_\ell} \ket{v}$, then there is $t\in J_F$ such that $k_t<C$.

\section{Other decidability results and lower bounds}\label{extensions}

For space reasons, we move the proof of Theorem~\ref{npthm} to the appendix. This theorem shows that determining if a cutpoint is isolated is at least NP-hard via an encoding of the subset sum problem (a well known NP-complete decision problem) into three state PFA.

In the remainder of this section we utilise Theorem~\ref{mainthm} to obtain some related decidability results. The first of these combines Theorem~\ref{mainthm} with a seminal result of Rabin and allows us to use our decidability result for cutpoint isolation to solve the emptiness problem for PFA on letter-bounded context-free languages when the cutpoint is isolated. We again highlight here that the emptiness problem is undecidable in general on letter-bounded languages, even when all matrices commute  and the PFA is polynomially ambiguous \cite{Bell19}.

\begin{cor}\label{newCor1}
The emptiness problem is decidable for probabilistic finite automata on letter-bounded context-free languages when the cutpoint is isolated.
\end{cor}
\begin{proof}
A seminal result of Rabin \cite{Ra63} showed that given a $n$-state PFA $\pfa$ acting on an alphabet $\Sigma$ and \emph{isolated} cutpoint $\lambda \in [0, 1]$ such that $\lambda$ is isolated by $\epsilon > 0$ (i.e. $|\pfa(w) - \lambda | > \epsilon$ for all $w \in \Sigma^*$), then there exists a DFA $\mathcal{D}$ such that $L_{< \lambda}(\pfa) = L(\mathcal{D})$, where $L(\mathcal{D})$ denotes the language accepted by the DFA $\mathcal{D}$. Moreover, Rabin showed that the number of states of $D$ is no more than $\left(1+\frac{|F|}{\epsilon}\right)^{n-1}$ where $F$ is the set of final states of $\pfa$.

We note that the proof of Theorem~\ref{mainthm} not only determines if a cutpoint is isolated but also determines an `isolation bound' $\epsilon > 0$ if it is isolated. In this case we can use Rabin's result to construct an equivalent DFA $\mathcal{D}_<$ recognising $L_{< \lambda}(\pfa)$. By inverting final and non final states of $\mathcal{D}_<$, we can construct $\mathcal{D}_{\geq}$ which recognises $L_{\geq \lambda}(\pfa)$. Finally we note that if $\lambda$ is isolated then $L_{< \lambda}(\pfa) = L_{\leq \lambda}(\pfa)$ and thus $\mathcal{D}_<$ and $\mathcal{D}_{\geq}$ recognise the same languages as $L_{\leq \lambda}(\pfa)$ and $ L_{> \lambda}(\pfa)$, respectively. Hence the emptiness problem is decidable.
\end{proof}

We now show that the value-$1$ problem for PFA on letter-bounded CFL inputs is decidable. This problem is undecidable for standard PFA but decidable for \#-cyclic automata \cite{GO10}.

\begin{cor}
The value-$1$ is decidable for probabilistic finite automata on letter-bounded context-free languages.
\end{cor}
\begin{proof}
This is trivial since the value-$1$ problem is equivalent to the isolation of the cutpoint $1$ for a PFA \cite{GO10}.
\end{proof}

Finally we note that Theorem~\ref{mainthm} trivially allows us to determine if a PFA which is allowed a fixed number of alternations between modes of operation (i.e.\ a fixed maximum number of alternations between input letters) has isolated cutpoints.

\begin{cor}
Given a probabilistic finite automaton $\pfa$ on alphabet $\Sigma = \{a_1, \ldots, a_\ell\}$, cutpoint $\lambda \in [0, 1]$ and maximum number $k > 0$ of alternations between input letters, then determining if the cutpoint is isolated is decidable.
\end{cor}
\begin{proof}
We may apply Algorithm~\ref{thealg} on $\pfa$ and $\lambda$ with each language from the following (finite) set of letter-bounded languages $\Lambda = \{w_1^* w_2^* \ldots w_k^* | w_i \in \Sigma\}$.

This defines the set of inputs where we alternate between the modes of operation a maximum of $k$ times (analogous to how counter automata models are often studied with a maximum number of alternations between increasing and decreasing the counters). If any $L \in \Lambda$ on Algorithm~\ref{thealg} returns that the cutpoint is not isolated then $\lambda$ is not isolated for $\pfa$ with a maximum number of alternations $k$, otherwise the cutpoint is isolated.
\end{proof}

\section{Conclusion}
In this work we showed that the cutpoint isolation problem is decidable for PFA when the input words are constrained to come from a letter-bounded  context-free language, even for exponentially ambiguous PFA. This is in contrast to the situation for the (strict) emptiness problem and the injectivity problem, which are \emph{undecidable} even over more restricted PFA for which all matrices commute, the PFA is polynomially ambiguous and the input words are over a simple letter-bounded language $a_1^*\cdots a_\ell^*$. We show that if the cutpoint is isolated for words over the input language, then the emptiness problem becomes decidable. We also show that the value-$1$ problem is decidable for these restricted input words.

It would be interesting to determine the complexity of the cutpoint isolation problem more precisely. We show an NP-hard lower bound in Theorem~\ref{npthm}. The algorithm we provide may belong to NP, however there are some issues with showing this upper bound, namely that the number of limits of a stochastic matrix may be exponential in its dimension and the value of constant $C$ from Proposition~\ref{limitlem} may be exponential in terms of the bit size of the matrices, making the verification stage of an NP algorithm difficult to achieve. Extending the results to more general bounded languages would also be an interesting future work.

\bibliographystyle{abbrv}
\bibliography{refs}

\newpage

\section{Appendix - proof of Theorem~\ref{npthm}}

We use a reduction from the subset sum problem, defined thus: given a set of positive integers $S = \{x_1, x_2, \ldots, x_k\} \subseteq \N$ and a natural number $T \in \N$, does there exist a subset $S' \subseteq S$ such that $\sum_{\ell \in S'} \ell = T$? This problem is well known to be NP-complete \cite{GJ79}. We define the set of matrices $M = \{A_i, B_i | 1 \leq i \leq k\} \subseteq \Q^{3 \times 3}$ in the following way:
$$
A_i = \frac{1}{x_i+1} \begin{pmatrix} 1 & x_i & 0 \\ 0 & 1 & x_i \\ 0 & 0 & x_i + 1 \end{pmatrix}, \quad B_i = \frac{1}{x_i+1} \begin{pmatrix} 1 & 0 & x_i \\ 0 & 1 & x_i \\ 0 & 0 & x_i + 1 \end{pmatrix}
$$
Note that $A_i$ and $B_i$ are thus row stochastic. Let $u = (1, 0, 0)^\top$ be the initial probability distribution, $v = (0,1,0)^\top$ be the final state vector and let $\pfa = (\bra{u}, \{A_i, B_i\}, \ket{v})$ be our PFA. We define the cutpoint $\lambda = \frac{T}{y}$, where $y = \sum_{j= 1}^k (x_j+1)$. Define letter-bounded language $\mathcal{L} = (a_1 | b_1) (a_2 | b_2) \cdots (a_k | b_k) \subseteq a_1^*b_1^*a_2^*b_2^*  \cdots a_k^*b_k^*$ and define a morphism $\varphi: \{a_i, b_i | 1 \leq i \leq k\}^* \to \{A_i, B_i | 1 \leq i \leq k\}^*$ in the natural way (e.g. the morphism induced by $\varphi(a_i) = A_i$ and $\varphi(b_i) = B_i$). Now, for a word $w = w_1w_2 \cdots w_k \in \mathcal{L}$, note that $w_j \in \{a_j, b_j\}$ for $1 \leq j \leq k$. Define that $\mathfrak{v}(a_i) = x_i$ and $\mathfrak{v}(b_i) = 0$ and inductively extend to $\mathfrak{v}:\Sigma^* \to \N$ by defining $\mathfrak{v}(w_1w_2 \cdots w_k) = \mathfrak{v}(w_1) + \mathfrak{v}(w_2 \cdots w_k)$ with $\mathfrak{v}(\varepsilon) = 0$. In this case, we see that (due to the structure of $A_i$ and $B_i$):
\[
\bra{u}\varphi(w_1 w_2 \cdots w_k)\ket{v} = \frac{ \mathfrak{v}(w)}{\sum_{j = 1}^{k}(x_j+1)} 
\] 
Note of course that the factor $\frac{1}{\sum_{j = 1}^{k}(x_j+1)}$ is the same for any $w \in \mathcal{L}$.

Assume that there exists a solution to the subset sum problem, i.e., there exists $S' \subseteq S$ such that $\sum_{\ell \in S'} \ell = T$. Then consider word $w = w_1w_2 \cdots w_k$ such that $w_j = a_j$ if $x_j \in S'$ and $w_j = b_j$ otherwise. In this case, $\sum_{i \in S'}^{k} x_i = \mathfrak{v}(w)$ and thus $\bra{u}\varphi(w_1 w_2 \cdots w_k)\ket{v} = \frac{T}{y} = \lambda$. If no solution exists, then for any word $w = w_1w_2 \cdots w_k$, $\left|\mathfrak{v}(w) - T\right| \geq 1$, and so $\left|\bra{u}\varphi(w_1 w_2 \cdots w_k)\ket{v} - \lambda\right| > \frac{1}{\sum_{j = 1}^{k}(x_j+1)}$ and thus $\lambda$ cannot be arbitrarily approximated.

Clearly the representation size of the PFA $\pfa$ and $\lambda$ are polynomial in the representation size of the subset sum problem instance and therefore we are done.

\end{document}